\documentclass[conference]{IEEEtran}
\usepackage{amsbsy}
\usepackage{floatflt} 

\usepackage{amsmath,stackrel}
\usepackage{amssymb}
\usepackage{times}
\usepackage{graphics}
\usepackage{graphicx}
\usepackage{xspace}
\usepackage{paralist} 
\usepackage{setspace} 
\usepackage{xypic}
\xyoption{curve}
\usepackage{latexsym}
\usepackage{theorem}
\usepackage{ifthen}
\usepackage{subfigure}
\usepackage{booktabs}
\usepackage{algorithm}
\usepackage{algorithmic}
\usepackage{color}

\newcounter{brojac}

{\theoremheaderfont{\it} \theorembodyfont{\rmfamily}
\newtheorem{remark}[brojac]{Remark}
}

{\theoremheaderfont{\it} \theorembodyfont{\rmfamily}
\newtheorem{theorem}{Theorem}
\newtheorem{lemma}[theorem]{Lemma}

}

\ifCLASSINFOpdf
\else
\fi
\begin{document}
%
\title{Slotted Aloha for Networked Base Stations}

\author{\IEEEauthorblockN{Dragana Bajovi\'c, Du$\check{\mbox{s}}$an Jakoveti\'c}
\IEEEauthorblockA{BioSense Center,\\
University of Novi Sad, 
Novi Sad, Serbia\\
Email: \{dbajovic, djakovet\}@uns.ac.rs}
\and
\IEEEauthorblockN{Dejan Vukobratovi\'c, Vladimir Crnojevi\'c}
\IEEEauthorblockA{Department of Power, Electronics, and Communications Engineering,\\
University of Novi Sad, 
Novi Sad, Serbia\\
Email:\{dejanv, crnojevic\}@uns.ac.rs}}


%


\maketitle

\vspace{-0mm}

\begin{abstract}
We study multiple base station, multi-access systems in which the user-base station adjacency is induced by geographical proximity. At each slot, each user transmits (is active) with a certain probability, independently of other users, and is heard
by all base stations within the distance~$r$. Both the users and base stations are
placed uniformly at random over the (unit) area. We first consider a non-cooperative decoding
where base stations work in isolation, but a user is decoded as soon as one of its nearby base stations reads
a clean signal from it. We find the decoding probability and quantify the gains
introduced by multiple base stations. Specifically, the peak throughput
increases linearly with the number of base stations~$m$ and is roughly $m/4$ larger than
the throughput of a single-base station that uses standard slotted Aloha. Next,
we propose a cooperative decoding, where the mutually close base stations inform each other whenever
they decode a user inside their coverage overlap. At each base station, the messages received from the nearby stations
help resolve collisions by the interference cancellation mechanism.
 Building from our exact formulas for the non-cooperative case, we provide a heuristic formula
  for the cooperative decoding probability that reflects well the actual performance.
  Finally, we demonstrate by simulation significant gains of cooperation with respect to the non-cooperative decoding.
\end{abstract}



%
\IEEEpeerreviewmaketitle

\vspace{-0mm}
\section{Introduction}
\label{section-intro}
Slotted {Aloha}~\cite{Roberts} and framed slotted {Aloha}~\cite{Fadra} are
well-known schemes for uncoordinated multiple access that date back to the 70s.
 With these schemes, the time is divided into slots, and, at each slot, users
contend to transmit their packets to the base station. With slotted Aloha,
each user transmits at each slot with a certain probability; with framed slotted Aloha, slots are grouped into frames, and each user randomly selects a slot at each frame to transmit.

\begin{figure}[thpb]
      \centering
      \includegraphics[trim = 0mm 6.5mm 0mm 0mm, clip, width=8cm]{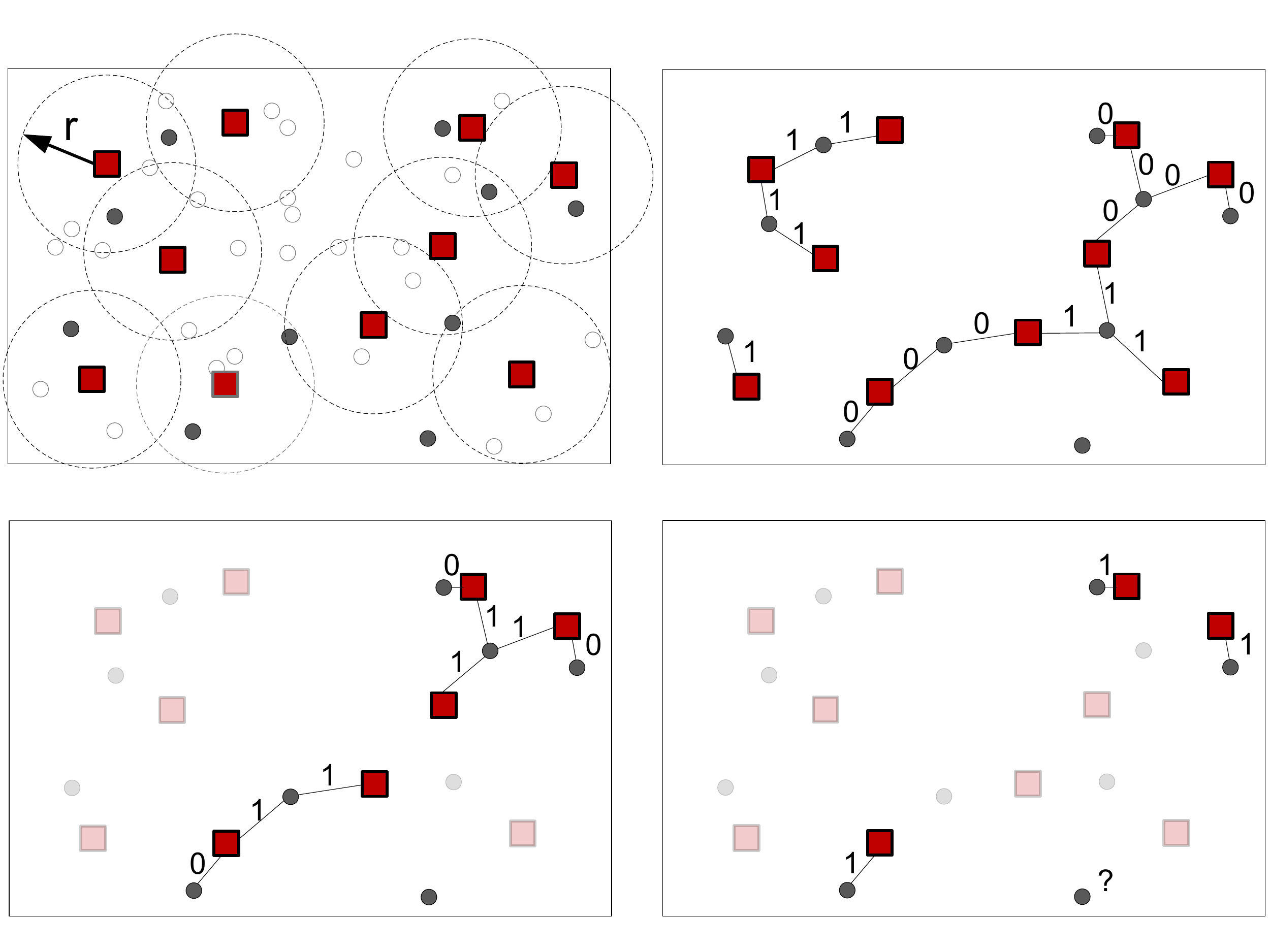} 
       \caption{Top left: Illustration of the multi-access system. Red squares represent base stations.
       Full circles represent active users, while empty circles represent inactive users.
       Top right, bottom left, and bottom right depict, respectively, first, second, and third iteration
       of the cooperative decoding algorithm for the network example in the top left figure.
       An edge is labeled with one at iteration~$t$ if its incident user gets collected at~$t$; such an
        edge is erased at~$t+1$; an edge is labeled with $0$ at iteration~$t$ if
        its incident user is unknown (not decoded) up to and at iteration~$t$.}
       \label{figure-system}
       \vspace{-.0cm}
\end{figure}

In the past decade, there has been much progress in the development of slotted Aloha type protocols,  e.g.,~\cite{SlottedALOHAwithIC,liva,FramelessALOHA}, with dramatic throughput improvements.
 Reference~\cite{SlottedALOHAwithIC} introduces a framed protocol
 that allows for multiple user transmissions, as considered in~\cite{DiversityAloha} in the past,
 but with a novel successive interference cancelation mechanism~\cite{SlottedALOHAwithIC}.
  In~\cite{SlottedALOHAwithIC}, users transmit (are active) at multiple slots at each frame,
 and send, along with their packet replicas, pointers to the corresponding activation slots.
 When a slot with a single active user occurs, this allows the base station not only collect this user,
  but also subtract its contribution in each other slot the user was active. This likely resolves collisions
 in some of the past slots and thereby allows for collecting additional users. More recently,~\cite{liva} demonstrates
  that successive interference cancelation is analogous to the belief propagation
  erasure-decoding of the codes on graphs. Exploiting this analogy,~\cite{liva} introduces a variable number
 of users' transmission attempts and optimizes their distribution to maximize the throughput.
 Building on an analogy with rateless codes, reference~\cite{FramelessALOHA} introduces the frameless Aloha protocol
  that further enhances the throughput. With the protocol therein, the frame size is not fixed a priori, but rather it adds new slots until a desired fraction of decoded
 users is achieved.

All the above references exploit \emph{temporal} diversity.
 Reference~\cite{ZorziSpatialDiversity} considers \emph{multiple receiver} multi-access systems with \emph{spatial} diversity which arises from independent fading of different user-receiver links.
 It analyzes the capture performance of the system under Rayleigh fading and shadowing.
%
%
 A recent reference~\cite{LivaNovo} also considers a multiple receiver case with spatial diversity.
  Under independent on-off fading, it quantifies analytically the gains in the throughput introduced by multiple receivers (over the single receiver case), as well as the impact of the fading probability on these gains.

In this paper, we also study the spatial diversity effects with multiple receivers,
 but under a very different model than the ones in~\cite{ZorziSpatialDiversity,LivaNovo}.
 A total of $m$ base stations (receivers) are deployed over a (unit) geographic area, and they
 jointly serve $n$ users (transmitters).
 Both the users and base stations are
 placed uniformly at random over the area. At a fixed time slot,
 each user transmits (is active) with probability~$p$, independently from other users.
 Each base station can hear all active users that are within distance~$r$ from it,
  where $r$ is small compared to the diameter of the area. The base station thus receives
 a superposition of the signals of active users in its $r$-neighborhood. (The signals of the users outside
 the $r$-neighborhood do not contribute to the signal.) 

We first consider the slotted Aloha protocol where each base station performs decoding in isolation (without
cooperating with other stations). It decodes a user whenever
there is a single active user in its $r$-neighborhood.
We find the probability that an arbitrary fixed user is decoded, both in the
finite regime and asymptotically, when $n,m\rightarrow \infty$ and $r \rightarrow 0$ ($p$-fixed.)
 Further, we quantify the gains of diversity introduced with multiple base stations.
 In particular, the peak throughput (expected number of decoded users per slot)
  is increased $\delta \, m$ times with respect to the single-base station slotted Aloha,
  where $\delta$ is a positive constant. (In particular,
  $\delta \approx 1/4$, see Section~\ref{section-performance-analysis} for details.)
  In other words, the throughput scales linearly with~$m$. For example, for $m=100$ and $r \approx 0.1$,
  the peak throughput is about~$20$.

Next, we propose a cooperative, iterative decoding where the base stations that are geographically
close communicate during decoding iterations. Specifically, we assume that, if a base station detects
a user, it knows at which other base stations this user is also heard, and it informs them of this users' ID and
its information packet. The contacted base stations can subtract the interference contribution of the received
signal, which possibly reveals additional clean packet readings.
We show by simulation that cooperation introduces significant gains
in the system performance. For example, for $m=100$ and $r \approx 0.1$, the
peak throughput increases from $20$ (no cooperation) to $33$ (with cooperation).
 Also, the maximal load for which the decoding probability is above a prescribed value (e.g.,
 $0.95$) is about $3$ times larger under cooperation than without cooperation,
 for a wide range of~$r$.

  Structurally, this decoding algorithm is analogous to the interference cancellation decoding
 in, e.g.,~\cite{liva,SlottedALOHAwithIC}, and it can be represented
 via message passing on a bipartite graph like in~\cite{liva}. Active users here correspond to users in~\cite{liva},
  base stations correspond to different slots (check nodes) in~\cite{liva},
  and the links are the physical links between active users and base stations.
  However, the structure of the graph here is induced by geometry and is very different from
  the random graph in~\cite{liva}. (See Section~\ref{section-decoding-algs} for details.)
   Evaluating the decoding probability here is very challenging and standard tools
   like and-or-tree analysis~\cite{AndOrTree} do not directly apply.
   We make the first step towards this goal by giving a heuristic formula
   that reflects well the actual performance. We derive the heuristic
   building from our results for the non-cooperative decoding.

In this paper, cooperation among base stations is confined \emph{within a single time slot} and is independent
across slots. In other words, this paper exploits \emph{spatial} diversity.
 In our ongoing work, we exploit the potential of both \emph{spatial and temporal} diversity
 by allowing that base stations cooperate both across space (as considered here) and across slots.
 The motivation for this comes from the single base station systems, where
 successive interference cancellation across slots yields dramatic throughput improvements.

Finally, we believe that our studies have a potential to find applications in massive uncoordinated multiple access in various networks, such as cellular, satellite, and vehicular networks, including recently popular machine-to-machine (M2M) services over these networks.

\textbf{Paper organization}. The next paragraph introduces notation.
Section~\ref{section-system-model} details the system model that we assume, and Section~\ref{section-decoding-algs}
presents our decoding algorithms. Section~\ref{section-performance-analysis} presents
our results on the performance of the two decoding algorithms. Section~\ref{section-numerical-interpretations}
gives numerical studies and interpretations. Finally, we conclude in Section~\ref{section-conclusion}.

\textbf{Notation}. We denote by: $\mathbf{B}(q,s)$
the Euclidean ball in the $2$-dimensional space centered at $q$ with radius $s$;
$\mathbf{B}_{\infty}(q,s)$ the square centered at $q$, with the side length equal to~$2 s$;
$1_{E}$ the indicator of event $E$; $\mathbb P$ and $\mathbb E$ the probability and expectation operators, respectively.

\vspace{-0mm}
\section{System model}
\label{section-system-model}
%
%

%
%
%
%
%
%
%
%
We consider a multi-access system with $n$ users and $m$ base stations. We denote by $U_i$, $i=1,...,n$ a user,
and by $B_l$, $l=1,...,m$, a base station. Users and base stations are distributed over a geographical area, and each user $U_i$ can be heard by all base stations within distance $r$ from~$U_i$. (See Figure~\ref{figure-system}, top left, for an illustration.) The time is divided into slots. As the number of users $n$ may be larger than the number of base stations $m$ (as is common in practical scenarios), to avoid excessive collisions, different users' transmissions are distributed across time slots, i.e., only a subset of users transmits at a certain slot.
 In this paper, we assume that decoding is completely decoupled (independent) across slots.
 Henceforth, from now on, it suffices to consider the system at a single, fixed slot. To keep the exposition general, we assume that each user $U_i$ transmits its message at a fixed slot with probability $p$, independently from other users, and that all transmissions are slot-synchronized. This model subsumes, e.g., the following system.
  There are $\tau$ available slots in each frame. Users' and
  base stations' placements are fixed during the frame. Each user transmits once per frame, with equal probability across the $\tau$ slots. In our model, this corresponds to setting $p=1/\tau$.
  We say that a user is active at a certain slot if it transmits at this slot.

 We let $G:=np/m$, and we call $G$ the normalized load. The quantity $G$ equals the expected
 number of active users at a fixed slot per base station.
 The message of user $U_i$ contains the information packet and a header with the user's ID.
 If $U_i$ is within distance $r$ from $B_l$, we say that $U_i$ and $B_l$ are adjacent.
 Each base station $B_l$ therefore hears a superposition (collided message in general) $y_l$ from all active adjacent users.
 We explain decoding mechanism in Section~\ref{section-decoding-algs}.

We now detail the placement model at a fixed slot. Both users and base stations at a fixed slot are placed in the unit square $\mathcal A=\mathbf B_{\infty}(0,1/2)$, centered at $(0,0)$. User $U_i$ is situated at a location $u_i$, where
 $u_i$ is selected from $\mathcal A$ uniformly at random, independently from other user's locations.
 Further, base station $B_l$ is positioned at a location $b_l$,
 where $b_l$ is selected from $\mathcal A$ uniformly at random, independently from other stations' locations.
 We assume that the placements of users and base stations are also mutually independent.


For the purpose of analysis, we differ two types of placements. We define $\mathcal A^{\mathrm{o},r}:= \mathbf B_{\infty}(0,1/2-2r)$, and say that a user is nominally placed if its position is in $\mathcal A^{\mathrm{o},r}$, and similarly for a base station. If, on the other hand, a user or a base station lies in the strip $\partial \mathcal A^{r}$ along the boundary of $\mathcal A$, $\partial \mathcal A^{r}:=\mathcal A\setminus \mathcal A^{\mathrm{o},r}$, we call this a boundary placement. Since placements are uniform over $\mathcal A$, the probability of the nominal placement is $(1-4r)^2$, and the probability of the boundary placement is $1-(1-4r)^2=8 r-16 r^2$. We see that, as the radius $r$ decreases, the probability of the nominal placement goes to one, and hence we can neglect in the analysis all the effects caused by the boundary placements.

\textbf{Degree distributions}. For future reference, we introduce the users' and base stations' degree distributions when they are
nominally placed. Fix arbitrary user $U_i$, and arbitrary point $q \in \mathcal{A}^{\mathrm{o},r}$, and let $\Lambda_d = \mathbb P \left( \mathrm{deg}(U_i)=d\,|\,u_i=q\right)$, i.e.,
 $\Lambda_d$ is the conditional probability that $U_i$ has exactly $d$ adjacent base stations, given that it is nominally placed.
  It is easy to show that degrees follow binomial distribution, i.e.,
 $
  \Lambda_d = {m \choose d} (r^2\pi)^d (1-r^2\pi)^{m-d},
 $
  $d=0,...,m$.
Similarly, let $\Psi_d = \mathbb P \left( \mathrm{deg}(B_l)=d\,|\,b_l=q \right)$,
where $\mathrm{deg}(B_l)$ denotes the number of \emph{active} users $U_j$, $j \in \{1,...,n\}\setminus \{i\}$, adjacent to $B_l$. (We exclude arbitrary fixed user $U_i$, as needed for subsequent analysis.) We have $
  \Psi_d = {n-1 \choose d} (p\,r^2\pi)^d (1-p\,r^2\pi)^{n-1-d}$, $d=0,...,n-1$.
We will also be interested in the asymptotic regime, when $n \rightarrow \infty$,
$r=r(n) \rightarrow 0$, and $m=m(n) \rightarrow \infty$ ($p$ is fixed), such that
%
%
 $m r^2 \pi \rightarrow \lambda,$ $n p\,r^2 \pi \rightarrow \psi,$
%
where $\lambda,\psi>0$ are constants.
In such setting, the users' and base stations' degree distributions
 converge to Poisson distributions with parameters $\lambda$ and $\psi$, respectively, i.e., for all $d=0,1,...$:
\begin{equation}
\label{eqn-lambda-psi-infty}
\Lambda_d  \rightarrow \Lambda_{\infty,d}:=e^{-\lambda}\frac{\lambda^d}{d!},\:\:\:\:
\Psi_d  \rightarrow \Psi_{\infty,d}:=e^{-\psi}\frac{\psi^d}{d!}.
\end{equation}
Hence, in the asymptotic regime, $\lambda$ is the average number of
base stations adjacent to a fixed user $U_i$, and $\psi$ is the
average number of active users adjacent to a fixed base station~$B_l$.
 It is easy to see that $\lambda $ and $\psi$ are related
 as $\psi = G\,\lambda$.

\textbf{Coverage}. Consider $\mathbb E\left[ \frac{1}{n} \sum_{i=1}^n 1_{\left\{U_i\mathrm{\;cov.}\right\}} \right]
=  \mathbb P \left( U_i\mathrm{\;cov.}\right)$,
where the event $\left\{ U_i\mathrm{\;cov.}\right\}$ means that $U_i$ is heard (``covered'') by at least
one base station. We refer to the latter quantity as the expected coverage.
%
 We have $P \left( U_i\mathrm{\;cov.}\,|\,u_i \in \mathcal{A}^{\mathrm{o},r}\right)=1-\Lambda_0$, and
$\mathbb P \left( U_i\mathrm{\;cov.}\right) \rightarrow 1-\Lambda_{\infty,0}=1-\mathrm{exp}(-\lambda)$.
 An active user $U_i$ can be collected
only if it is covered, no matter what decoding is used. Therefore, for a high decoding probability, we cannot have
$\lambda$ (or $r$) too small. Henceforth, from now on
we assume $\lambda \geq \lambda_{\mathrm{min}}(\epsilon) := \mathrm{ln}(1/\epsilon)$,
such that $1-\epsilon$ coverage is ensured; e.g., for $\epsilon=0.05$, $\lambda_{\mathrm{min}}(\epsilon) \approx 3$.


\vspace{-0mm}

\section{Decoding algorithms}
\label{section-decoding-algs}
Subsection~\ref{subsection-noncooperative} details the non-cooperative decoding, and Subsection~\ref{subsection-cooperative}
 details the cooperative decoding algorithm.
\vspace{-0mm}
\subsection{Non-cooperative decoding}
\label{subsection-noncooperative}
We now explain the non-cooperative decoding algorithm, where each base station works in isolation.
 At each base station, decoding is the simple slotted Aloha decoding. Suppose that station $B_l$
  received signal $y_l$. We assume that $B_l$ can determine if $y_l$ corresponds to a ``clean'' message.
  In other words, if, at a fixed slot, there is a single active user $U_i$ in $\mathbf{B}(b_l,r)$
   then $B_l$ collects user $U_i$ (it reads its packet and obtains its ID).
   We say that a user is collected at a fixed time slot if it is collected by at least
   one base station at this slot.
   For example, for the network in Figure~\ref{figure-system}, top left,
   we can see that $4$ out of $10$ active users are collected.
%
%
%
\subsection{Cooperative decoding}
\label{subsection-cooperative}
We now present the cooperative decoding algorithm, where neighboring base station collaborate to collect users.
We assume that each base station $B_l$ is aware of which users (either active or inactive) it covers, i.e., it knows the IDs
 of all its adjacent users (e.g., through some sort of association procedure). Further, for each of its adjacent users $U_i$, $B_l$ knows the list of the base stations $B_k$, $k \neq l$, to which $U_i$ is also adjacent.
%
%
 We say that two base stations are neighbors if
they share at least one user. The decoding is iterative and involves communication between
neighboring base stations. Each base station $B_l$ maintains over iterations $t$, $t=0,1,...,$ a signal $z_l=z_l(t)$.
Initially, at $t=0$, $z_l(t)$ is the received signal $y_l$ from its active adjacent users (either a clean message from an active user, a collided message, or an empty message if neither of the users in $\mathbf{B}(b_l,r)$ is active.)
Station $B_l$ at a certain iteration~$t$ may receive
a message $x^{(k)}$ from a neighboring base station $B_k$.
This happens if $B_k$ decodes a user at~$t$, which we call~$U^{(k)}$,
and if $U^{(k)}$ is adjacent to both $B_l$ and $B_k$.  The message $x^{(k)}$ contains the packet of
user $U^{(k)}$ and its ID.
 Upon reception of $x^{(k)}$, station $B_l$ subtracts the interference contribution of user $U^{(k)}$, which we symbolically write as $z_l \leftarrow z_l-x^{(k)}$. Station $B_l$ can
recognize if the updated signal $z_l$ corresponds to a clean packet, and, if so, it reads the packet and determines to which user it belongs.\footnote{Our decoding puts an additional physical requirement on the receivers. To illustrate this, let users' messages $x_i$, $i=1,...,n$, be real, positive numbers and signal $y_l$
received by $B_l$ be $y_l=\sum_{j \in \Omega_l}h_{l,j} x_j$. Here, $h_{l,j}$ is the (positive)
 channel gain, and $\Omega_l$ is the set of active users covered by~$B_l$.
 Let, at the first decoding iteration~$t=1$, $B_k$ reads a clean signal
 $y_k = h_{k,i}\,x_i$, where $U_i$ is adjacent to both $B_l$ and $B_k$.
  Then, $B_k$ transmits to $B_l$ the message $x^{(k)}=x_i=y_k/h_{k,i}$. Upon reading the real number $x_i$,
 $B_l$ performs the (real-number) subtraction
 $z_l \leftarrow  z_l -h_{l,i} \,x_i$. Note that $B_k$ and $B_l$ need to perform re-scaling
 by their respective channel gains. Hence, each $B_l$ needs the channel gains $h_{l,j}$ to all its adjacent users~$U_j$.
 Now, consider~\cite{liva}, where
  each slot (check node) $l$ lies at the same physical location of the single base station.
 Check node $l$ receives
 $y_l=\sum_{j \in O_l} h_j \,x_j$, where $h_j$ is the channel gain from
 user $U_j$ to the base station, and $O_l$ are active users at slot~$l$. When check node $k$ has a clean signal
 $y_k = h_{i}\,x_i$ with $U_i$ adjacent to both check nodes (different slots) $l$ and $k$,
 it can just ``send'' $y_k$ to $l$, and $l$
 performs $z_l \leftarrow z_l - y_k$. Hence, no re-scaling by the channel gains $h_i$'s is needed.} The decoding algorithm operates as follows. At iteration $t$, $t=1,2,...$, all base stations work in synchrony and perform the same steps.
%
%
%
 Iteration~$t$ at an arbitrary station~$B_l$ has three steps: 1) check signal,
 2) collect and transmit, and 3) receive and update. The last two steps always occur exclusively, i.e., one and only
 one among the two is always performed. As we will see, each base station performs at most $m$ iterations. We assume
 that all stations know $m$ beforehand.

\emph{Step~1: Check signal}: $B_l$ checks whether signal $z_l$ corresponds to a ``clean'' packet. If this is true, it performs
 the collect and transmit step; otherwise, it performs the receive and update step.

\emph{Step~2: Collect and transmit}: $B_l$ collects a user $U^{(l)}$ and reads its ID.
 It transmits message $x^{(l)}$ to all $B_k$'s, $k \neq l$, that are adjacent to $U^{(l)}$.
We call the latter set of stations $\Omega^{(l)}.$ After transmissions, $B_l$ leaves the algorithm.

%

\emph{Step~3: Receive and update}: $B_l$ scans over all messages $x^{(k)}$ that it received at~$t$
  and identifies the subset~$\mathcal{J}^{(l)}$ of all distinct messages.\footnote{Among the received messages, there may be repetitions, i.e., there may be two or more
 equal messages received.} Subsequently, $B_l$
  subtracts from $z_l$ the interference contributions from all $x_j$'s,  $j \in \mathcal{J}^{(l)}$, which
  we symbolically write as $z_l \leftarrow z_l - \sum_{j \in \mathcal{J}^{(l)}}x_j$.
  Set $t \leftarrow t+1$. If $t=m$, $B_l$ leaves the algorithm; otherwise, it
  goes to step~1.


\textbf{Graph representation of decoding}. We now introduce a graphical message-passing representation of
decoding. It involves the evolution of a bipartite graph $\mathcal G_t$ over iterations $t$.
Graph $\mathcal{G}_t$ has two types of nodes -- base stations and active users. Both the node sets and the edge set change (reduce) over iterations~$t$.
 It is initialized by $\mathcal{G}_0$, where $\mathcal{G}_0$ is defined as follows: it has the node set
 that consists of all base stations and all \emph{active} users. Its set of links contains all pairs $(B_l,U_i)$ such that
$B_l$ and $U_i$ are within distance $r$ from each other (and $U_i$ is active.)
We now describe one iteration~$t$.

\emph{Graph decoding iteration}. All $B_l$'s in $\mathcal{G}_t$
 check in parallel if their degree in $\mathcal G_t$ equals one. Let $\mathcal{L}_t \subset \{1,...,m\}$ be the set of degree one base stations in $\mathcal G_t$. If $\mathcal L_t$ is empty, the algorithm terminates. Otherwise,
 for each $l \in \mathcal L_t$, let $U^{(l)}$ be the user adjacent to $B_l$.
  Remove from $\mathcal G_t$ all $B_l$'s and $U^{(l)}$'s, $l \in \mathcal L_t$, and all the links incident to $U^{(l)}$, $l \in \mathcal L_t$.
 Set $t \leftarrow t+1$. 

It is easy to see that the above algorithm terminates after at most $m$ iterations. Namely, at each iteration $t$, either
at least one base station node is removed, or the algorithm terminates at~$t$. Therefore, at most $m$ iterations can be performed.
 For the network in Figure~\ref{figure-system}, top left,
 we show decoding iterations in Figures~\ref{figure-system}, top right (at $t=1$),
 bottom left ($t=2$), and bottom right ($t=3$). We can see that cooperative
 decoding collects $9$ out of $10$ users, while the non-cooperative collected~$4$.
\footnote{
The graph decoding algorithm here is very similar to
the interference cancelation decoding in, e.g.,~\cite{liva}, and the
iterative (graph-peeling) decoding of {LDPC} codes over the erasure channel, e.g.,~\cite{ERC}.
 The analogy with~\cite{liva} is that base stations here correspond
to different slots (check nodes) in~\cite{liva}, and \emph{active} users here
correspond to users in~\cite{liva}.
However, the graph structure here
is very different from the one in~\cite{liva}. First, for two users $U_i$ and $U_j$ with $u_i$ close to $u_j$, there is
a high overlap between $\mathbf{B}(u_i,r)$ and $\mathbf{B}(u_j,r)$, and hence
 the sets of their adjacent base stations (the check nodes
 to which $U_i$ and $U_j$ connect) have a high overlap.
 This is in contrast with the random graph model where the neighborhood sets of different users
 are independent. Second, in~\cite{liva}, with high probability, the sizes of cycles
 grow with $n$ as $\log n$, while here small cycles occur with a non-vanishing probability.}

%
%

%
%
%
%
%
%

\vspace{-0mm}

\section{Performance analysis}
\label{section-performance-analysis}
In this section, we study the performance of both non-cooperative and cooperative decoding schemes.
Specifically, our goal is to determine the expected \emph{fraction} of decoded users per time slot,
$
\mathbb E\left[ \frac{1}{n} \sum_{i=1}^n 1_{\left\{U_i\mathrm{\;coll.}\right\}} \right].
$
 Exploiting the symmetry across users, we have that the above quantity
equals $(1/n) \,n\,\mathbb P \left( U_i\mathrm{\;coll.}\right)$ =  $\mathbb P \left( U_i\mathrm{\;coll.}\right). $
Hence, our task reduces to finding the probability that  arbitrary fixed
user $U_i$ is collected.
The following simple relation will be useful throughout:
$\mathbb P \left( U_i\mathrm{\;coll.}\right) = p\,\mathbb P \left( U_i\mathrm{\;coll.}\,|\,U_i\mathrm{\;act.}\right)$, which is easily
obtained after conditioning on the event that $U_i$ is active and using that $p$ is the probability of $U_i$ being active.
(Here, abbreviation ``$U_i$ coll.'' stands for $U_i$ is collected, and
``$U_i$  act.'' stands for $U_i$ is active.) We also consider the normalized, per station throughput
$T(G)=(1/m)\mathbb E\left[ \sum_{i=1}^n 1_{\left\{U_i\mathrm{\;coll.}\right\}} \right]$ -- the expected \emph{total} number of collected users per slot, per station. Next, recall $\lambda=m r^2 \pi$, and, for fixed $m,p,$ and $r$, we will be interested in the following quantity:
\begin{figure}[thpb]
\vspace{-0mm}
      \centering
       \includegraphics[trim = 0mm 7mm 0mm 10mm, clip, height=2.5 in,width=3.7 in]{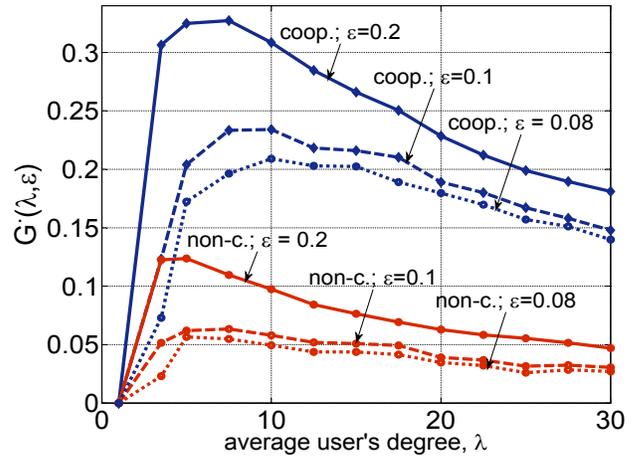}
       \caption{Simulated quantity ${G^\bullet(\lambda,\epsilon)}$ in~(2) versus the
       average user's degree $\lambda = m r^2\pi$ for different values of $\epsilon \in \{0.08;0.1;0.2\}$.}
       \label{figure-G-bullet}
       \vspace{-0mm}
\end{figure}
\begin{equation}
\label{eqn-metric}
{G^\bullet(\lambda,\epsilon)}= \sup \left\{ G \geq 0:\,\, \mathbb P \left( U_i\mathrm{\;coll.}\,|\,U_i\mathrm{\;act.}\right) \geq 1-\epsilon \right\},
\end{equation}
where $\epsilon>0$ is a small number.
In words, ${G^\bullet(\lambda,\epsilon)}$ is the largest normalized load for which decoding probability
$ \mathbb P \left( U_i\mathrm{\;coll.}\,|\,U_i\mathrm{\;act.}\right) $ is above the prescribed value~$1-\epsilon$.
 Recall $P\left( U_i\,\,\mathrm{cov.}\right)$ and that $P\left( U_i\,\,\mathrm{cov.}\right)\approx 1-\Lambda_0$.
 It is clear that, when $1-\epsilon > P\left( U_i\,\,\mathrm{cov.}\right)$, due to the relation
 $\mathbb P \left( U_i\mathrm{\;coll.}\,|\,U_i\mathrm{\;act.}\right) \leq P\left( U_i\,\,\mathrm{cov.}\right)$,
  $\mathbb P \left( U_i\mathrm{\;coll.}\,|\,U_i\mathrm{\;act.}\right)$ cannot be
 greater or equal $1-\epsilon$ for any $G$, i.e., no matter how small $G$ is.
 Thus, whenever $1-\epsilon > P\left( U_i\,\,\mathrm{cov.}\right)$, by convention
 we say ${G^\bullet(\lambda,\epsilon)}=0.$

\begin{remark}
We explain the motivation behind quantity~${G^\bullet(\lambda,\epsilon)}$.
 Suppose there are $\tau$ available slots, where each user
 is active in exactly one among the slots. The system has the following requirement on the ``quality of service''
 -- each user~$U_i$ be collected with probability above $1-\epsilon$.
 This translates into the requirement
 $\mathbb P \left( U_i\mathrm{\;coll.}\,|\,U_i\mathrm{\;act.}\right) \geq 1-\epsilon.$
  For a fixed $m,p,r$, we ask what is the maximal number of users
  that can be served with the guaranteed quality of service. That is,
   we look for $\sup\{n:\,\mathbb P \left( U_i\mathrm{\;coll.}\,|\,U_i\mathrm{\;act.}\right) \geq 1-\epsilon\}$.
   As $m, p, r$, are fixed and $G=n p/m$, this is equivalent to finding~\eqref{eqn-metric}. We will later
   be interested in optimizing (maximizing) ${G^\bullet(\lambda,\epsilon)}$.
   \end{remark}

\vspace{-0mm}

\subsection{Non-cooperative decoding}
\label{subsection-non-coop-analysis}
We now characterize $\mathbb P \left( U_i\mathrm{\;coll.}\right)$ for the non-cooperative decoding.
 As we will see, the sought probability depends on the distributions of the areas covered by randomly generated balls.
%
%
  Specifically, consider the ball~$\mathbf{B}(0,r)$. Fix some $k \geq 1$, and generate randomly $k$ points $q_1,...,q_k$, where
 $q_l$'s, $l=1,...,k$, are drawn mutually independently from the uniform distribution on~$\mathbf{B}(0,r)$,
 and let $\alpha_k$ be the random variable that equals the area of $\cup_{l=1,...,k}\mathbf{B}(q_l,r)$ divided (normalized) by $r^2\pi$.
 Further, denote by $\mu_k$ the probability distribution on $[0,\infty)$ induced by $\alpha_k$.
  Clearly, $\alpha_k$ (and hence $\mu_k$) does not depend on $r$ due to normalization.
  Hence, we can set $r=1$. Also, it is easy to see that $\alpha_1 = 1$ with probability one, i.e.,
  $\mu_1$ is the Dirac distribution at $1$. Also, for any $k$, $1 \leq \alpha_k \leq 4$, with probability one, i.e.,
  $\mu_k$ is supported on $[1,4]$. This is because all the $q_l$'s, $l=1,...,k$, belong to
  $\mathbf{B}(0,1)$, and thus $\cup_{l=1,...,k}\mathbf{B}(q_l,1)$ is always a subset
  of $\mathbf{B}(0,2)$.   The distributions $\mu_k$, $k=2,...,m$, are difficult to compute.
  However, they can be partially characterized
  by estimating
  the first $s_{\mathrm{max}}$ moments $\overline{\alpha}_k^{(s)}:=\int_{1}^4
  a^s \,d\mu_k(a)$, $s=1,...,s_{\mathrm{max}}$. This can be done, e.g.,
   through Monte Carlo simulations.
   We emphasize that the moments $\overline{\alpha}_k^{(s)}$, $s=1,...,s_{\mathrm{max}}$,
   $k=2,...,m$, need to be tabulated only once (just like, e.g.,
   the tail distribution of the standard Gaussian.) That is, once we have the $\overline{\alpha}_k^{(s)}$'s available,
   they apply for any set of parameters~$n,m,p,r$.


  We now state our result on $\mathbb P \left( U_i\mathrm{\;coll.}\,|\,U_i\mathrm{\;act.}\right)$ and $T(G)$.
  We distinguish two cases: 1) non-asymptotic regime of finite $r,n,m$, that corresponds to (binomial) degree distributions $\Lambda_d,\Psi_d$; and 2) asymptotic regime
  that corresponds to (Poisson) degree distributions $\Lambda_{\infty,d}$, $\Psi_{\infty,d}$.
\begin{theorem}
\label{theorem-non-coop}
Consider the non-cooperative decoding algorithm. Then, for $r \leq 1/4$, we have:
$
p\,P_{\mathrm{coll.}}^{\,\mathrm{o},r}\leq \mathbb P \left( U_i\mathrm{\;coll.}\right)\leq p\,\left(P_{\mathrm{coll.}}^{\,\mathrm{o},r} +8r-16r^2\right),
$
 where
 $P_{\,\mathrm{coll.}}^{\mathrm{o},r} = \mathbb P \left(U_i\mathrm{\;coll.}\,|\,U_i\mathrm{\;act.},\,u_i \in \mathcal{A}^{\mathrm{o},r}\right)$ and equals:
\begin{equation}
\label{eqn-P-zeta-k}
P_{\,\mathrm{coll.}}^{\mathrm{o},r} = \sum_{k=1}^m (-1)^{k-1}\,\zeta_k\,{I_k},\;\;\;\zeta_k= \sum_{d=k}^m {d \choose k}\,\Lambda_d,
\end{equation}
and
$
{I_k}= \int_{a=1}^4 \left(1-p\, r^2\pi a \right)^{n-1}\,d\mu_k(a).
$
 Further, let $p$ be fixed, $n \rightarrow \infty$, $m = m(n) \rightarrow \infty$, and $r=r(n) \rightarrow 0$,
and recall $\lambda,\psi$ in~\eqref{eqn-lambda-psi-infty}. Then,
$
\mathbb P \left( U_i\mathrm{\;coll.} \right) \rightarrow
p\,\sum_{k=1}^{\infty} (-1)^{k-1}\,\frac{\lambda^k}{k!}\,{I_{\infty,k}},
$
where
$
{I_{\infty,k}}= \int_{a=1}^4 e^{-\psi\,a}\,d\mu_k(a).
$
\end{theorem}
We briefly comment on the structure of the results. It can be seen that $\zeta_k \rightarrow \lambda^k/k!$. This is because
 $\Lambda_d \rightarrow \Lambda_{\infty,d}$ in~\eqref{eqn-lambda-psi-infty}, and so
 $\zeta_k \rightarrow \sum_{d=k}^{\infty} {d \choose k}e^{-\lambda}\lambda^d/d!
  = \sum_{d=k}^{\infty} \frac{d!}{(d-k)! \,k!}\,e^{-\lambda}\frac{\lambda^{d}}{d!} =
  \lambda^k/k!$. Similarly, it can be shown that ${I_k} \rightarrow {I_{\infty,k}}$.
%

\emph{Sketch of the proof}. The detailed technical proof of Theorem~\ref{theorem-non-coop} is omitted due to lack of space and will be
provided in a companion journal paper. We briefly sketch the proof of~\eqref{eqn-P-zeta-k}, highlighting the main steps and omitting
 certain arguments. The keys are to use the inclusion-exclusion principle, conditioning on a user's degree, and
 considering the size of the areas covered by the user's neighboring base stations.
 We consider a user $U_i$ at the fixed nominal placement $u_i=q$, $ q \in \mathcal{A}^{\mathrm{o},r}$. Consider $\mathbb P \left( U_i\mathrm{\;coll.}\,|\,U_i\mathrm{\;act.},\,u_i=q\right)$. We first use the total probability law with respect to the degree $\mathrm{deg}(U_i)$:
\begin{align}
\label{eqn-proof-1}
\mathbb P \left( U_i\mathrm{\;coll.}\,|\,U_i\mathrm{\;act.},\,u_i = q \right) =
 \sum_{d=1}^m \gamma(d)\,\Lambda_d,
\end{align}
 where $\gamma(d) = \mathbb P\left( U_i\mathrm{\,coll.}\,|\, U_i\mathrm{\,act.,}\,U_i=q\,\,\mathrm{deg}(U_i)=d\right)$. Due to base stations' symmetry, without loss of generality we can assume that $B_1,...,B_d$ are the
neighbors. In other words, $\gamma(d)$ actually equals the probability that $U_i$ is collected, given that $U_i$ is active, $U_i$ is nominally placed at $q$, and $B_1,...,B_d$ are its neighbors.
Now, $U_i$ is collected if and only if at least one of the $B_l$'s decodes it, and $B_l$ collects $U_i$ if and only if
 $U_i$ is the only active user in $\mathbf{B}(b_l,r)$. (We then say that $B_l$ is empty.)
 Summarizing, $\gamma(d)$ is the conditional probability of the union
 $\cup_{l=1}^d\{B_l\mathrm{\,\mathrm{empty}}\}$, given
 that the $U_i$'s neighbors are $B_1,...,B_d$, $u_i=q$, $U_i$-active. Now, applying
 the inclusion-exclusion formula:
 \begin{align}
 \label{eqn-proof-2}
 \gamma(d) = \sum_{k=1}^d (-1)^{k-1} {d \choose k} \eta(k,d),
 \end{align}
 where $\eta(k,d)$ is the conditional probability of
 $\cap_{l=1}^k \{B_l\mathrm{\;empty}\}$,
 conditioned on $B_1,...,B_d$ be neighbors, $u_i=q$, $U_i$-active.
 Intuitively, $\eta(k,d)$ depends on the area
 of $\cup_{l=1}^k \mathbf{B}(b_l,r)$, because
 we look at the event that no active users lie in~$\cup_{l=1}^k \mathbf{B}(b_l,r)$.
 It can be shown (proof omitted) that $\eta(k):=\eta(k,d)$ equals~$I_k$ in Theorem~\ref{theorem-non-coop} and  is
 independent of~$d$.
  Substituting $\eta(k,d)=I_k$ in~\eqref{eqn-proof-2}, plugging the resulting
  equation in~\eqref{eqn-proof-1}, and noting that the probability in~\eqref{eqn-proof-1} is the same for all
  nominal $q$'s, we obtain~\eqref{eqn-P-zeta-k}.

\textbf{Numerical calculation}.
Theorem~\ref{theorem-non-coop} expresses $P \left( U_i\mathrm{\;coll.} \right)$ in the form that is
difficult to compute. Assuming that the first moments $\overline{\alpha}_k:=\overline{\alpha}_k^{(1)}$,
$k=1,...,k_{\mathrm{max}}$ are available (e.g., obtained through Monte Carlo simulations), we can compute $\mathbb P \left( U_i\mathrm{\;coll.} \right)$ with a high accuracy
and a small computational cost, through the formula:
\begin{equation}
\label{eqn-P-zeta-k-approx}
\mathbb P \left( U_i\mathrm{\;coll.} \right) \approx
p\,\sum_{k=1}^{k_{\mathrm{max}}} (-1)^{k-1}\,\frac{\lambda^k}{k!}\,e^{-\overline{\alpha}_k\,\psi}.
\end{equation}
Here, we approximated $\sum_{k=1}^\infty(-1)^{k-1}\lambda^k/k! I_{\infty,k}$ at $k=k_{\mathrm{max}}$
by letting
 ${I_{\infty,k}}\approx e^{-\overline{\alpha}_k\,\psi}$ and truncating
 the infinite sum at $k=k_{\mathrm{max}}$.
 Formula~\eqref{eqn-P-zeta-k-approx} gives high accuracies for $m$ of order $100$ or larger.
 Given the quantity $k_{\mathrm{max}}$, approximation is
 accurate for $\lambda$ sufficiently smaller than $k_{\mathrm{max}}$, e.g.,
 $\lambda\leq 0.25 \cdot k_{\mathrm{max}}$. (In other words, for a larger $\lambda$, larger
 $k_{\mathrm{max}}$ is needed.)\footnote{More generally,
 for arbitrary set of system parameters $n,m,p,r$, formulas in~\eqref{eqn-P-zeta-k} and~\eqref{eqn-P-zeta-k-approx}
 are computable with arbitrarily high accuracy, provided that $k=k_{\mathrm{max}}$
 is large enough and the moments $\overline{\alpha}_k^{(s)}$, $k=1,...,k_{\mathrm{max}}$, $s=1,...,s_{\mathrm{max}}$,
  are available. The accuracy can be controlled with the increase of $k_{\mathrm{max}}$ and $s_{\mathrm{max}}$.
  For example, consider ${I_k}$. The integrand $(1-p\,a r^2\pi)^{n-1}$ is a polynomial
of order $n-1$ and can be written as $\sum_{s=0}^{n-1} c_s\,a^s$, where $c_s=c_s(n,p,r)$ are the coefficients.
Therefore, ${I_k} = \sum_{s=0}^{n-1} c_s \overline{\alpha}_k^{(s)}$.}

\textbf{A simple lower bound on $\mathbb P(U_i\mathrm{\,coll.})$}.
We derive a lower bound on $\mathbb P(U_i\mathrm{\,coll.})$, which is loose but very useful in providing insights into the system performance.
We exploit this bound in Section~\ref{section-numerical-interpretations}.
\begin{lemma}
\label{lemma-non-coop}
Consider the non-cooperative decoding algorithm in the asymptotic
setting as in the second part of Theorem~\ref{theorem-non-coop}. Then:
$
\lim_{n \rightarrow \infty}\mathbb P \left( U_i\mathrm{\,coll.}\right) \geq p\,(1-e^{-\lambda})e^{-4 \psi}.
$
\end{lemma}
\begin{proof} Consider $U_i$ at a nominal placement $q \in {\mathcal A}^{\mathrm{o},r}$ and suppose that $U_i$ is active.
If there exists a base station in $\mathbf{B}(q,r)$ and there are no active users in $\mathbf{B}(q,2r)$ (other than $U_i$), then $U_i$ is collected. Let $\widehat {P}$ denote the probability of the former; clearly, $\mathbb P \left( U_i\mathrm{\,coll.}|U_i\mathrm{\,act.},\,u_i=q\right) \geq \widehat P$. By the independence of the users' and base stations' placements, we have that $\widehat P = (1 - (1 - r^2 \pi)^m)(1 - 4 r^2 \pi)^{n-1}$, which in the asymptotic regime goes to $(1- e^{-\lambda}) e^{-4\psi}$. Passing to the limit (where boundary effects vanish), the result
follows.
\end{proof}


\vspace{-0mm}

\subsection{A heuristic for cooperative decoding}
\label{subsection-coop-analysis}
We now derive a heuristic formula for $\mathbb P \left( U_i\mathrm{\;coll.} \right)$
with cooperative decoding. The heuristic relies on
our arguments for the non-cooperative case. It takes into account only the first two iterations of cooperative decoding.
 (See Remark~\ref{remark-heuristic}.)
Consider arbitrary fixed user~$U_j$ at a nominal placement.
Let $1-\sigma_1$ be the probability that
 $U_j$ has been collected after the first iteration~$t=1$, given that it is active.
   It is easy to see that this is precisely the corresponding decoding probability
 for the non-cooperative case. We thereby approximate it with
  $1-\sigma_1 \approx \sum_{k=1}^{k_{\mathrm{max}}}(-1)^{k-1}\lambda^k/k! \, \mathrm{exp}(-\overline{\alpha}_k\,\psi)$.
   Now, let $1-\sigma_2$ be the probability that
 arbitrary fixed user has been collected after $2$ iterations,
 given it is active. We take a conservative approach by
 approximating the decoding probability after complete decoding algorithm
 be~$1-\sigma_{2}$. We now evaluate $1-\sigma_2$. Fix user~$U_i$.
 We neglect the boundary effects and
consider $U_i$ at the nominal placement $u_i=q$, $q \in {\mathcal A}^{\mathrm{o},r}$, i.e., we set $1-\sigma_2 \approx \mathbb P \left( U_i\mathrm{\;coll.}\,|\,U_i\mathrm{\;act.},\,u_i = q \right). $
Using the total probability law with respect to the $U_i$'s degree:
$
1-\sigma_2 \approx
 \sum_{d=1}^m \gamma^\prime(d)\,\Lambda_d,
$
 where $\gamma^\prime(d) = \mathbb{P}(U_i\mathrm{\,coll.}\,|\,U_i\mathrm{\,act.},\,u_i=q,\,
 \mathrm{deg}(U_i)=d)$.
 Fix~$d$, and without loss of generality, fix the neighborhood of $U_i$ to $B_1,...,B_d$.
  For each $B_l$, $l=1,...d$, let $1-\rho_1$ be the probability
  that all active users~$U_j$, $j \neq i$, adjacent to
  $B_l$, have been decoded after iteration~$t=1$. (In the graph representation of decoding, this
  corresponds to $B_l$ being
  connected only to $U_i$ after iteration~$t=1$.)
  We say in the latter case that $B_l$ is known.
  Note that $U_i$ has been collected after $t=2$ if and only if
  there exists at least one $B_l$, $l=1,...,d$, such that $B_l$ is known after $t=1$ (We write this shortly  as ``$B_l$ known''.)
   If the graph were random as in, e.g.,~\cite{liva}, the events $\{B_l\mathrm{\,known}\}$ would be
   mutually independent, and, moreover,
   they would be independent of $\mathrm{deg}(U_i)=d.$ Then, we would have the following formula:
   $\gamma^\prime(d)=1-\rho_1^d$. However, this is not the case here, and
   we need to proceed in a different way. In particular,
   we account for
   the correlation of the events $\{B_l\mathrm{\,known}\}$ through the inclusion-exclusion formula on the event $\cup_{l=1}^d \mathbb{P}(B_l\mathrm{\,known})$:
 $
 \gamma^\prime(d) = \sum_{k=1}^d (-1)^{k-1} {d \choose k} \eta^\prime(k,d).
$
 Here, $\eta^\prime(k,d)$ is the conditional probability of
 $\cap_{l=1}^k \{B_l\mathrm{\;known}\}$,
 given that $B_1,...,B_d$ be the neighbors of~$U_i$, $u_i=q$, $U_i$-active.
  It remains to find $\eta^\prime(k,d)$. Clearly, this quantity
   depends on the area
 of $\cup_{l=1}^k \mathbf{B}(b_l,r)$, but also it depends on~$d$.
  We now approximate~$\eta^\prime(k,d)$ by accounting for the former dependence and by neglecting the latter dependence.
  Specifically,
  for each $k,d$,
  we let
  \begin{equation}
  \label{eqn-approx}
  \eta^\prime(k,d) \approx \int_{a=1}^4 (1-\rho_1)^{a} d\mu_k(a) \approx (1-\rho_1)^{\overline{\alpha}_k}.
  \end{equation}
  \begin{remark}
  \label{remark-heuristic}
  The motivation for~\eqref{eqn-approx} is the following.
  Suppose that, after~$t=1$, the unknown users inside $\mathbf{B}(q,2r)$ followed
  a Poisson distribution with mean $\lambda\,\sigma_1$. (On average,
  there are $a \lambda$ users over an area of size~$a$, the fraction of which
  are unknown is $\sigma_1.$) Further, suppose that
  their distribution does not depend on~$d$ -- the number of
  base stations in $\mathbf{B}(q,r)$. (This is not
  the case in general, as more base stations in~$\mathbf{B}(q,r)$ tend
  to reduce the number of unknown users.)
  Then, the probability that
  there are no unknown users in $\mathbf{B}(b_l,r)$ would be
  $1-\rho_1 = \mathrm{exp}(-\lambda\,\sigma_1\,r^2\pi)$. Recall $\alpha_k$ from Section~\ref{section-performance-analysis}. Then,
  given that the area of $\cup_{l=1}^k \mathbf{B}(b_l,r)$ is $r^2\pi \alpha_k$,
  the probability that there are no unknown users in $\cup_{l=1}^k \mathbf{B}(b_l,r)$ would be
  $\nu_k = \mathrm{exp}(-\lambda\,\sigma_1 \alpha_k\, r^2 \pi)$.
  Therefore, $\nu_k = (1-\rho)^{\alpha_k}$; averaging
  with respect to $\alpha_k$, we finally obtain the left approximation~\eqref{eqn-approx}. (The right one
  is as in~\eqref{eqn-P-zeta-k-approx}.) The dependence of~$\eta^\prime(k,d)$ on~$d$ is more pronounced
  when the considered iteration~$t$ is larger (in~\eqref{eqn-approx}, it is $t=1$),
  and~\eqref{eqn-approx} becomes less accurate. Thus, we stop at $t=2$ and let~$
  \mathbb P (U_i\mathrm{\,coll.}) \approx p(1-\sigma_2)$. Our future work will address the dependence of $\eta^\prime(k,d)$ on~$d$.
  \end{remark}
  Now, proceeding analogously to the non-cooperative case, and taking
  the asymptotic setting, we obtain:
 \begin{equation}
 \label{eqn-approx-resenje-1}
  \sigma_2
=1-
\sum_{k=1}^{k_{\mathrm{max}}}(-1)^{k-1} \frac{\lambda^k}{k!} (1-\rho_1)^{\overline{\alpha}_k}.
\end{equation}
Formula~\eqref{eqn-approx-resenje-1} can be seen as a counterpart
of the following formula from the density evolution analysis on random graphs: $\sigma_2 = 1 - \mathrm{exp}(-\lambda\,\rho_1)$.
It remains to express $\rho_1$ in terms of $\sigma_1$.
We omit details due to lack of space, but
the derivation of the approximate formula is completely
dual (analogous) to~\eqref{eqn-approx-resenje-1}. A fixed station~$B_l$
 is replaced with fixed user~$U_i$, the total probability law is done
 with respect to $\Psi_d$ instead of $\Lambda_d$, and the event
 $\cup_{l=1}^d\{B_l\mathrm{\;known}\}$ is replaced with $\cup_{j=1}^d \{U_j\mathrm{unknown}\}$.
 The resulting formula is:
\begin{align}
\label{eqn-approx-resenje-2}
\rho_1 = \sum_{k=1}^{k_{\mathrm{max}}}(-1)^{k-1} \frac{\psi^k}{k!} (\sigma_1)^{\overline{\alpha}_k}.
\end{align}
In summary, we set $\mathbb{P}(U_i\mathrm{\,coll.}) \approx p(1-\sigma_2)$,
where $\sigma_2$ is in~\eqref{eqn-approx-resenje-1}, $\rho_1$ is in~\eqref{eqn-approx-resenje-2}, and
 $\sigma_1 = 1 - \sum_{k=1}^{k_{\mathrm{max}}}(-1)^{k-1} \frac{\lambda^k}{k!} \mathrm{exp}(-\overline{\alpha}_k\,\psi).$

\section{Numerical studies and interpretations}
\label{section-numerical-interpretations}
In this section, we carry out numerical studies and provide interpretations of our results.

\textbf{Simulation setup}. We explain the simulation setup that we use throughout the section. We set the number of base stations $m=100$, and
users' activation probability $p=0.25$. We simulate the decoding probability $\mathbb P(U_i\mathrm{\,coll.}\,|\,U_i\mathrm{\,act.})$
 and the normalized throughput $T(G)$ for different values of $G=n p/m$ by varying $n$.
 We vary $n$ such that $G$ varies within the interval $[0,1]$. We evaluate
 these quantities through Monte Carlo simulations. For each value of $n$,
 we generate one instance of the network, i.e., we place users and base stations
 uniformly over a unit square. For a fixed network, we run the
 non-cooperative and cooperative decoding. (With cooperative decoding, we
 simulate its graph representation.) For each fixed~$n$,
 we perform $1000$ simulation runs ($1000$ different graphs and the decoding algorithms over them.)
   For each $n$ (each~$G$), we
    estimate $\mathbb P(U_i\mathrm{\,coll.}\,|\,U_i\mathrm{\,act.})$ as $\widehat{P}/p$,
    where $\widehat{P}$ is the estimate of the decoding probability,
    and equals the total number of collected users divided by~$n$.
     We estimate the normalized throughput $T(G)$ as the total number of collected users divided by~$m$.
     We obtain the parameters $\overline{\alpha}_k$, $k=1,...,k_{\mathrm{max}}$,
     through Monte Carlo simulations.
     For each $k$, we repeat $4000$ different random, uniform, placements of $k$ points $q_k$
     in $\mathbf{B}(0,1)$. For each placement~$s$, $s=1,...,4000$, we estimate
     the area $a_s$ of $\cup_{l=1}^k \mathbf{B}(q_l,r)$ through the
     Monte Carlo simulation with $30,000$ trials. We set $k_{\mathrm{max}}=34.$
      We estimate $\overline{\alpha_k}$ as $(1/\pi)(1/4000)\sum_{s=1}^{4000} a_s$.
       With the non-cooperative decoding, we
       evaluate $\mathbb P(U_i\mathrm{\,coll.}\,|\,U_i\mathrm{\,act.})$ via~\eqref{eqn-P-zeta-k-approx},
       and $T(G) = G \,\mathbb P(U_i\mathrm{\,coll.}\,|\,U_i\mathrm{\,act.})$.
        With the cooperative decoding, we evaluate
        $\mathbb P(U_i\mathrm{\,coll.}\,|\,U_i\mathrm{\,act.})$ via~\eqref{eqn-approx-resenje-1}--\eqref{eqn-approx-resenje-2},
       and $T(G) = G \,\mathbb P(U_i\mathrm{\,coll.}\,|\,U_i\mathrm{\,act.})$.

\textbf{Throughput}. In the first experiment, we simulate
$T(G)$
 versus $G$ for $\lambda_{m,r}=m r^2\pi = 3$, and $\lambda_{m,r}=6$.
 Both values of $\lambda_{m,r}$ ensure coverage of at least $0.95$. Figure~\ref{figure-throughput} plots
 $T(G)$ for the non-cooperative and
 cooperative decoding. We depict both
 the theoretical values and the values obtained through simulations.
We first assess our theoretical findings. We can see that, in the non-cooperative case, our formula~\eqref{eqn-P-zeta-k-approx}
 accurately matches simulation. For the larger value $\lambda_{m,r}=6$, there is a slight
 mismatch, which could be eliminated by taking a larger~$k_{\mathrm{max}}$.
 For the cooperative case, our heuristic formulas~\eqref{eqn-approx-resenje-1}--\eqref{eqn-approx-resenje-2}
 follow well the trend of the curves, and we see that the heuristic
 is more accurate for the smaller value~$\lambda_{m,r}=3$. 
  Next, we compare the two decoding algorithms. From Figure~\ref{figure-throughput}, we can see
 that cooperation produces significant gains with respect to the non-cooperative decoding.
 For example, for $\lambda_{m,r}=3$,
  the peak throughput under cooperation is about $0.33$, while
  without cooperation it is below $0.2$. Similarly,
  for $\lambda_{m,r}=6$, the peak throughput
  under cooperation is $0.29$, while
  without cooperation we have about~$0.13$.

\textbf{Comparisons with a single base station}. To quantify the gains of diversity induced by multiple base stations, we
also compare the two decoding methods with the standard slotted Aloha and a
single base station. Suppose that we have one
base station at the center of the region, and that the station covers the full region.
Clearly, for such system and a large $n$,
$\mathbb P(U_i\mathrm{\,coll.}\,|\,U_i\mathrm{\,act.}) \approx \mathrm{exp}(-n\,p)$,
and the throughput is $n\,p\,\mathrm{exp}(-n\,p)$.
The peak throughput is $1/e\approx 0.37.$
To compare the single base station system with
the above two decoding algorithms,
we evaluate for each of the two the \emph{un-normalized} peak throughput.
For~$\lambda_{m,r}=3$, this quantity equals $0.33 \times m=33$ for cooperation,
 and $20$, without cooperation. Hence, $m=100$ base stations
 allow at least $54$ times larger throughput. More generally, consider the system in the asymptotic setting,
with $\lambda \geq \lambda_{\mathrm{min}}(\epsilon)=\mathrm{ln}(1/\epsilon)$ (i.e., with
the $1-\epsilon$ coverage.) From Lemma~\ref{lemma-non-coop},
with non-cooperative decoding (and hence with cooperative as well) $T(G)
\geq T^\prime(G) = G (1-e^{-\lambda})e^{-4\,G\lambda}$, where we used that $\psi=G \lambda$.
 For a fixed $\lambda$, the peak of $T^\prime(G)$ is
 $1/(4 \,e\,\lambda)(1-\mathrm{exp}(-\lambda))$.
 The maximum over $\lambda \geq \mathrm{ln}(1/\epsilon)$
  is attained at~$\mathrm{ln}(1/\epsilon)$ and equals
  $\frac{1}{4 \,e}\frac{1-\epsilon}{\mathrm{ln}(1/\epsilon)}$.
  Hence, comparing with the single base station system,
  the $m$-base station systems
  gives at least $\frac{1}{4}\frac{1-\epsilon}{\mathrm{ln}(1/\epsilon)} \times m$
   higher total (un-normalized) throughput.
   In other words, the total throughput grows linearly with the number of base stations~$m$.

\textbf{Quantity $G^\bullet(\epsilon,\lambda)$ and optimal radius~$r$}. We now give insights into how the system performance depends on $r$
and $\lambda$ for the two decoding algorithms, and we demonstrate large gains
of cooperation. Figure~\ref{figure-G-bullet} depicts
simulated $G^\bullet(\lambda,\epsilon)$ for $\epsilon \in \{0.2;0.1;0.08\}$.
First, we can see that, for each considered $\epsilon$,
 cooperation offers almost three times better (larger)
  $G^\bullet(\lambda,\epsilon)$ in a wide range of~$\lambda$.
  Second, we can see that there is an optimal~$\lambda(\epsilon)$.
  Consider, e.g., the non-cooperative case.
  On one hand, too small $r$ does not allow for
  sufficient coverage, and hence it yields poor performance.
   On the other hand, too large $r$ eliminates the benefits of diversity.
   To see this, just consider the case where $r=1$ and each base station covers all users.
   In this case, all base stations have same observations, and we effectively have the single base station system.

\vspace{-0mm}

\section{Conclusion}
\label{section-conclusion}
We studied effects of spatial diversity and cooperation of slotted Aloha
protocols with multiple base stations. Users and base stations are deployed
uniformly at random over a unit area. At a fixed slot, each user
transmits its packet (is active) with probability~$p$ and is heard
by all base stations placed within distance~$r$ from it.
 We first considered the non-cooperative decoding where
 a user is collected if it is a single active user at
 one of the base stations that hear it. We find the decoding probability
  and quantify the gains with respect to the standard single base station slotted Aloha.
  We show that the peak throughput with $m$ base stations is roughly
  $m/4$ times larger than when a single base station is available.
  Next, we propose a cooperative decoding, where the nearby base stations
  help each other resolve users' collisions through the interference cancellation mechanism.
    We demonstrate by simulation significant gains of cooperation
    with respect to the non-cooperative decoding. For example, for $m=100$ and $r \approx 0.1$,
    the peak throughout increases from $20$ to $33$. Also, the
    maximal load ($=n p/m$) for which the decoding probability is above $0.95$
    increases three times, for a wide range of~$r$. Finally, we
    give a heuristic formula for the decoding probability under cooperation. The formula
    accounts for the problem geometry and reflects well the actual performance.

\vspace{-0mm}

\begin{figure}[thpb]
      \centering
      \includegraphics[trim = 0mm 18mm 0mm 16mm, clip, height=1.8 in,width=3.5 in]{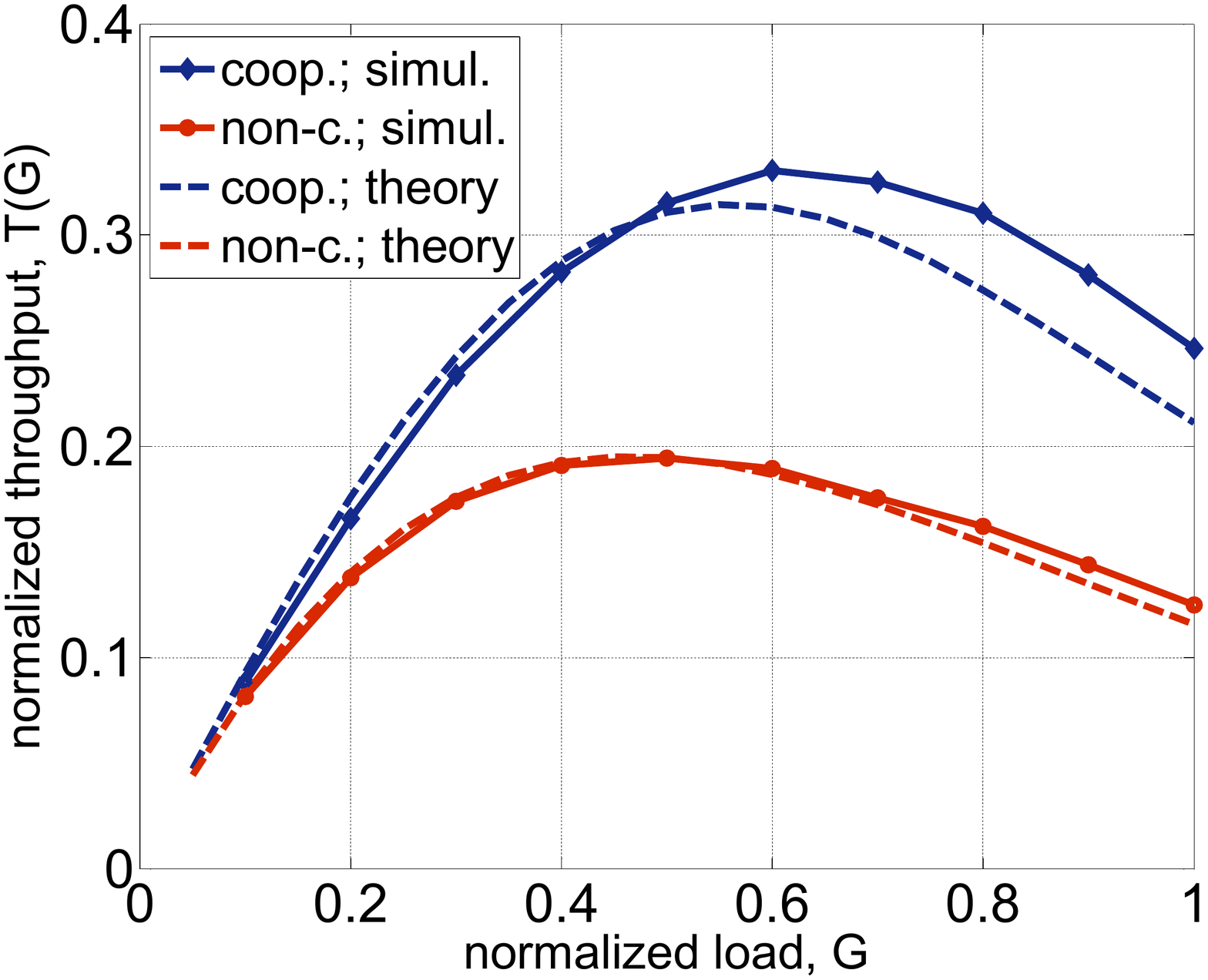} 
       \includegraphics[height=2.1 in,width=3.5 in]{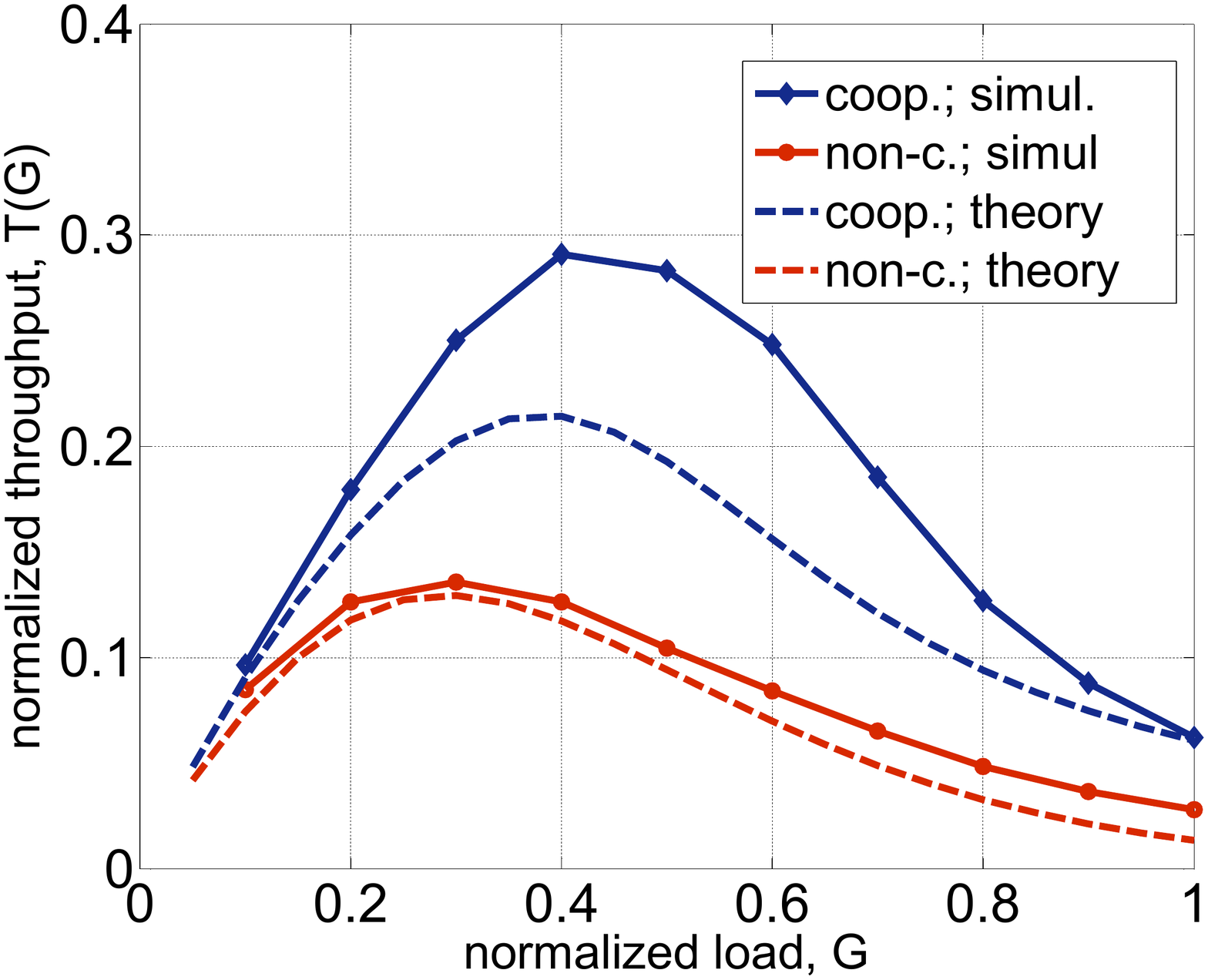}
       \vspace{-0mm}
       \caption{Normalized throughput $T(G)$ versus the normalized load $G=n p/m$, for the
       average user's degree $\lambda=3$ (top) and $\lambda=6$ (bottom).}
       \label{figure-throughput}
\end{figure}

\vspace{-0mm}



%
%
%

\vspace{3mm}

\bibliographystyle{IEEEtran}
\bibliography{IEEEabrv,bibliography}

\end{document}